\documentclass[conference]{IEEEtran}

\usepackage{times,latexsym,amsfonts,amsmath,amssymb,mathrsfs,verbatim,cite}
\usepackage{epsfig,epsf}
\usepackage{graphicx}
\usepackage[dvips]{color}
\usepackage{txfonts}

\def\nN{{\mathbb N}}
\def\zZ{{\mathbb Z}}

\def\eE{{\mathbb E}}
\def\pP{{\mathbb P}}

\def\QED{\mbox{\rule[0pt]{1.5ex}{1.5ex}}}

\def\@begintheorem#1#2{\tmpitemindent\itemindent\topsep 0pt\rm\trivlist
    \item[\hskip \labelsep{\indent\it #1\ #2:}]\itemindent\tmpitemindent}
\def\@opargbegintheorem#1#2#3{\tmpitemindent\itemindent\topsep 0pt\rm \trivlist
    \item[\hskip\labelsep{\indent\it #1\ #2\
    \rm(#3):}]\itemindent\tmpitemindent}
\def\@endtheorem{\endtrivlist\unskip}

\newtheorem{theorem}{Theorem}
\newtheorem{definition}{Definition}
\newtheorem{fact}{Fact}

\newtheorem{lemma}{Lemma}

\newtheorem{remark}{Remark}
\renewcommand{\theequation}{\arabic{section}.\arabic{equation}}

\begin{document}

\title{Efficient File Synchronization: a Distributed Source Coding Approach$^{\text{\small 1}}$}


\author{\IEEEauthorblockN{Nan Ma, Kannan Ramchandran and David Tse}
\IEEEauthorblockA{Wireless Foundations, Dept. of Electrical Engineering and Computer Sciences\\
University of California at Berkeley} }

\maketitle
\begin{abstract}
The problem of reconstructing a source sequence with the presence of
decoder side-information that is mis-synchronized to the source due
to deletions is studied in a distributed source coding framework.
Motivated by practical applications, the deletion process is assumed
to be bursty and is modeled by a Markov chain. The minimum rate
needed to reconstruct the source sequence with high probability is
characterized in terms of an information theoretic expression, which
is interpreted as the amount of information of the deleted content
and the locations of deletions, subtracting ``nature's secret'',
that is, the uncertainty of the locations given the source and
side-information. For small bursty deletion probability, the
asymptotic expansion of the minimum rate is computed.
\end{abstract}
\section{Introduction}
\addtocounter{footnote}{+1} \footnotetext{This material is based
upon work supported by the US National Science Foundation (NSF)
under grants 23287 and 30149 and by a gift from Qualcomm Inc.. Any
opinions, findings, and conclusions or recommendations expressed in
this material are those of the authors and do not necessarily
reflect the views of the NSF.}


In distributed file backup or file sharing systems, different source
nodes may have different versions of the same file differing by a
small number of edits including deletions and insertions. The edits
usually appear in bursts, for example, a paragraph of text is
deleted, or several consecutive frames of video are inserted. An
important question is: how to efficiently send a file to a remote
node that has a different version of it? Further, what is the
fundamental limit of the number of bits that needs to be sent to
achieve this goal?

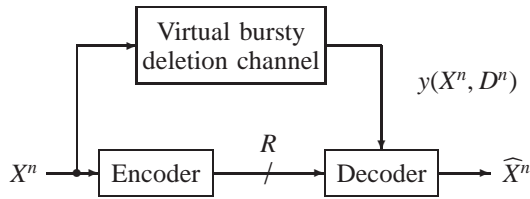
\begin{figure}[!htb]
\centering
\begin{picture}(9,2.5) 
\put(1,0){
    \put(1,0){\framebox(1.5,0.6){Encoder}}
    \put(4,0){\framebox(1.5,0.6){Decoder}}
    \put(1.5,1.5){\framebox(2.5,1){$\mbox{Virtual bursty} \atop \mbox{deletion channel}$}}
    \put(0.3,0.3){\vector(1,0){0.7}}
    \put(0.7,0.3){\circle*{.1}}
    \put(0.7,0.3){\line(0,1){1.7}}
    \put(0.7,2){\vector(1,0){0.8}}
    \put(4,2){\line(1,0){0.75}}
    \put(4.75,2){\vector(0,-1){1.4}}
    \put(5.5,0.3){\vector(1,0){0.7}}
    \put(2.5,0.3){\vector(1,0){1.5}}
    \put(-0,0.3){\makebox(0,0){$X^n$}}
    \put(6.55,0.35){\makebox(0,0){$\widehat {X}^n$}}
    \put(5.90,1.5){\makebox(0,0){$y(X^n,D^n)$}}
    \put(3.25,0.3){\makebox(0,0){/}}
    \put(3.25,0.7){\makebox(0,0){$R$}}
    }
\end{picture}
\caption{\label{fig:system} \small \sl Synchronizing source
sequences based on deletion side-information}
\end{figure}

In this paper, we study the problem of reconstructing a source
sequence with the help of decoder side-information using a
distributed source coding framework (see Figure~\ref{fig:system} for
an illustration of the system). In this paper we focus on a simple
case where the side-information is a deleted version of the source
sequence. Consider a binary sequence of length $n$ denoted by $X^n =
(X_1,\ldots,X_n)$. Consider another binary sequence of length $n$
called deletion pattern, denoted by $D^n = (D_1,\ldots,D_n)$, which
determines how $X^n$ is to be deleted. The outcome of the deletion
process, denoted by $y(X^n, D^n)$, is derived from $X^n$ by deleting
the bits at those locations where the deletion pattern is $1$. Here
is an example:
\begin{eqnarray*}
X^n &=& (0, 1, 0, 1, 1, 0, 1, 0, 1, 0) \\
D^n &=& (0, 1, 1, 0, 0, 0, 1, 1, 1, 0) \\
y(X^n,D^n)&=& (0, 1, 1, 0, 0).
\end{eqnarray*}
Note that the deletion pattern $D^n$ tends to have bursts of
consecutive $1$'s, which lead to bursty deletions. The original
files $X^n$ and the deleted files $y(X^n, D^n)$ are available to the
encoder and the decoder, respectively. The encoder sends a message
to the decoder, so that the latter can reconstruct (synchronize) the
original files $X^n$ with an error probability that is vanishing
when $n$ goes to infinity. The objective of this work is to
characterize the minimum rate of the message defined as the minimum
number of bits per source bit.



The problem of synchronizing edited sequences has been studied by
\cite{Leven, OrlitskyDeletion} under the assumptions (1) the decoder
is not allowed to make any error, and (2) the number of edits is a
constant that does not increase with the length of the sequence.
Upper and lower bounds of the minimum number of communication bits
were provided as functions of the number of edits and the length of
the sequence. In \cite{venkataramanan-interactive}, an interactive,
low-complexity and asymptotically optimal scheme was proposed. In
comparison, in this paper, we consider on information theoretic
formulation allowing a positive probability of error that vanishes
as $n$ increases. This assumption allows us to use additional
techniques like random binning to improve the minimum rate. Unlike
in assumption (2), we consider the case that a vanishing fraction of
source bits, rather than a constant number of bits, is deleted, to
get which makes the problem harder and more realistic.

In this paper, we characterize the minimum rate in terms of the
limit of the conditional entropy of the source sequence given the
side-information. We interpret the minimum rate as the amount of
information in the deleted content and the locations of the
deletions, subtracting the uncertainty of the locations given the
source and side-information. We refer to the latter as ``nature's
secret''. This is the information that the decoder will never find
out even if it knows the source sequence and the side-information
exactly; it represents the over-counting of information in the
locations of the deletions. For example, if $X^n=(0,0)$ and
$y(X^n,D^n)=(0)$, the decoder will never know and never needs to
know whether the first bit or the second bit is deleted. Therefore
the information about the precise location of the deleted bit is
over-counted and should be subtracted. For small deletion rate and
geometrically distributed burst length, the minimum rate is computed
up to the precision of two leading terms.

If the deletion pattern $D^n$ is independent and identically
distributed (iid), $X^n$ and $y(X^n,D^n)$ are the input and output
of a binary iid deletion channel (see \cite{MitzSurvey} and
references therein). In this case, the problem of characterizing the
minimum rate to reconstruct iid uniform source sequences in the
distributed source coding problem is closely related to the
evaluation of the mutual information across the deletion channel
with iid uniform input distribution. For small deletion probability,
the second and third order terms\footnote{For small deletion
probability $d$, the first order term of the channel capacity is 1,
the second order term is $\Theta(d \log d)$, and the third order
term is $\Theta(d)$.} of the channel capacity are achieved by iid
uniform input distribution and are computed in
\cite[Lemma~III.1]{MontanariISIT10}. In this paper we consider the
asymptotic expansion of the minimum rate for the general bursty
deletion process where the deletions are correlated over time. In
the special case of iid deletion process, the expansion in
Theorem~\ref{thm:smallbeta} reduces to
\cite[Lemma~III.1]{MontanariISIT10}. Note that in the source coding
problem, the constant term becomes zero, which means that the second
and third order terms of the channel capacity correspond to the
first and second order terms of the minimum rate. Therefore,
although it is mathematically equivalent to evaluate the these terms
for the source coding and channel coding problems, from the
practical point of view, the evaluation is more important for the
source coding problem than for the channel coding problem. See
Remark~\ref{rem:Montanari} for detailed discussions.

When we generalize the iid deletion process to bursty deletion
process, new techniques are introduced. The most interesting
technique is the generalization of the usual concept of a ``run''.
We view the sequence $(1,0,1,0,1,0)$ as a run with respect to
deletion bursts of length two, because deleting two consecutive bits
from that sequence always results in the same outcome sequence
$(1,0,1,0)$.

The rest of this paper is organized as follows. In
Section~\ref{sec:problem} we formally setup the problem and provide
a preview of the main result. In Section~\ref{sec:general} we
provide information theoretic expressions of the minimum rate for
general parameters of the deletion pattern. In
Section~\ref{sec:smallbeta} we focus on the asymptotics when the
deletion rate is small and compute the two leading terms of the
minimum rate. All the proofs are provided in the appendices.

{\it Notation:} With the exception of the symbols $R, E, C$, and
$J$, random quantities are denoted in upper case and their specific
instantiations in lower case. For $i, j\in \zZ$, $V_i^j$ denotes the
sequence $(V_i, \ldots, V_j)$ and $V^i$ denotes $V_1^i$. The binary
entropy function is denoted by $h_2(\cdot)$. All logarithms are base
2. The notation $\{0,1\}^n$ denotes the $n$-fold Cartesian product
of $\{0,1\}$, and $\{0,1\}^*$ denotes $\left(\bigcup_{k\in \zZ^+}
\{0,1\}^k\right) \bigcup \{\emptyset\}$.

\section{Problem Formulation and Main Result}
\label{sec:problem}


\subsection{Problem formulation}
The source sequence $X^n = (X_1,\ldots,X_n)\in \{0,1\}^n$ is iid
Bernoulli$(1/2)$. Let $\alpha, \beta \in (0,1)$. The deletion
pattern $(D_0, D_1, \ldots,D_{n+1})$ is a two-state stationary
Markov chain illustrated in Figure~\ref{fig:markov} with the initial
distribution $p_{D_0} \sim$ Bernoulli$(d)$, where
$d:=\beta/(\alpha+\beta)$ and transition probabilities
$\pP(D_i=0|D_{i-1}=1)=1-\pP(D_i=1|D_{i-1}=1)=\alpha$ and
$\pP(D_i=1|D_{i-1}=0)=1-\pP(D_i=0|D_{i-1}=0)=\beta$, for all
$i=1,2,\ldots,n+1$. Note that the initial distribution $p_{D_0}$ is
the stationary distribution of the Markov chain. The deleted
sequence $y(X^n, D^n)\in \{0,1\}^*$ is a subsequence of $X^n$, which
is derived from $X^n$ by deleting all those $X_i$'s with $D_i=1$
\footnote{$D_0$ and $D_{n+1}$ do not determine the deletion of any
source bit and do not play a role in the problem formulation.
However, they are used in the information theoretic expressions in
Sections~\ref{sec:general} and \ref{sec:smallbeta}.}. The length of
$y(X^n, D^n)$, denoted by $L_y$, is a random variable taking values
in $\{0,1,\ldots,n\}$. For $i < L_y$, $Y_i$ denotes the $i$-th bit
in the $y(X^n, D^n)$ sequence. A run of consecutive $1$'s in the
deletion pattern is called a burst of deletion. Since $\beta$ is the
probability to initiate a burst of deletion, it is called the
deletion rate.

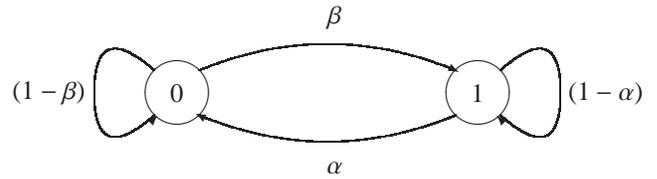
\begin{figure}[!htb]
\centering
\begin{picture}(9,2.5) 
\put(1.3,1){
    \put(1,0){\circle{0.8}}
    \put(1,0){\makebox(0,0){$0$}}
    \put(5,0){\circle{0.8}}
    \put(5,0){\makebox(0,0){$1$}}
    \put(0,0){\qbezier(1.3,0.3)(3,1)(4.7,0.3)}
    \put(4.4,0.45){\vector(2,-1){0.35}}
    \put(0,0){\qbezier(1.3,-0.3)(3,-1)(4.7,-0.3)}
    \put(1.6,-0.45){\vector(-2,1){0.35}}
    \put(0,0){\qbezier(0.7,0.3)(-0.1,1)(-0.1,0)}
    \put(0,0){\qbezier(0.7,-0.3)(-0.1,-1)(-0.1,0)}
    \put(0.4,-0.65){\vector(1,1){0.35}}
    \put(0,0){\qbezier(5.3,0.3)(6.1,1)(6.1,0)}
    \put(0,0){\qbezier(5.3,-0.3)(6.1,-1)(6.1,0)}
    \put(5.6,-0.65){\vector(-1,1){0.35}}
    \put(3.1,1){\makebox(0,0){$\beta$}}
    \put(3.1,-1){\makebox(0,0){$\alpha$}}
    \put(-0.7,0){\makebox(0,0){$(1-\beta)$}}
    \put(6.7,0){\makebox(0,0){$(1-\alpha)$}}
    }
\end{picture}
\caption{\label{fig:markov} \small \sl Markov model for the deletion
pattern process $\{D_i\}_{i\geq 0}$. $D_i=1$ means $X_i$ is deleted;
$D_i=0$ means $X_i$ is not deleted.}
\end{figure}


The source sequence $X^n$ is available to the encoder and the
deleted sequence $y(X^n, D^n)$ is available only to the decoder as
side-information.
The deletion patterns $D^n$ is available to neither the encoder nor
the decoder. The encoder encodes $X^n$ and sends a message to the
decoder so that the decoder can reproduce the source with high
probability.

\begin{remark} If $\beta = 1-\alpha=d$, $D^n$ becomes
iid, and the relation between $X^n$ and $y(X^n, D^n)$ can be modeled
as an iid deletion channel with deletion probability $d$. In this
paper we consider the Markov deletion pattern to emphasize the
bursty nature of the deletion process in the source coding problem.
\end{remark}

The formal definitions of a code and an achievable rate are as
follows.

\begin{definition}\label{def:code}
A distributed source code for deletion side-information with
parameters $(n,|{\mathcal M_n}|)$ is the tuple $(f_n,g_n)$
consisting of an encoding function $f_n: \{0,1\}^{n} \rightarrow
\mathcal M_n$ and a decoding function $g_n: \mathcal M_n \times
\{0,1\}^* \rightarrow \{0,1\}^{n}$.
\end{definition}

\begin{definition}\label{def:achievablerate}
A real number $R$ is called an achievable rate if, there exists a
sequence of distributed source codes $\{(f_n,g_n)\}_{n\geq 1}$ for
deletion side-information with parameters $(n,|{\mathcal M_n}|)$
satisfying $\lim_{n\rightarrow \infty} \pP(X^n \neq
g_n(f_n(X^n),y(X^n,D^n))) = 0$ and $\limsup_{n\rightarrow \infty}
(1/n)\log |{\mathcal M_n}|\leq R$.
\end{definition}

The set of all achievable rates is necessarily closed and hence the
minimum exists. The minimum achievable rate is denoted by $R_{min}$.
The focus of this paper is to characterize $R_{min}$, especially for
small $\beta$.

%

\subsection{Main result}
In Section~\ref{sec:general} we express $R_{min}$ using information
theoretic quantities when the parameters $\alpha$ and $\beta$ take
arbitrary values. Unfortunately, we cannot provide an explicit
expression of $R_{min}$ as a function of $\alpha$ and $\beta$. Hence
we focus on asymptotic regimes in Section~\ref{sec:smallbeta} when
$\beta$ is small.

Since the main difference between the erasure process and the
deletion process is that the locations of the erasures are explicit
but those of the deletions are not, it is interesting to focus on a
regime where the amount of information to describe the locations of
the deletions should play a significant role in the minimum rate.
When $\alpha$ is vanishing and the length of bursts of deletions is
increasing, for each burst, the number of bits to describe the
deleted content increases linearly with respect to the length of the
burst, but the number of bits to describe the location and length of
the burst increases logarithmly. Therefore the regime with a
vanishing $\alpha$ is not interesting. On the contrary, when
$\alpha$ is fixed, the length of a burst is of order $\Theta(1)$ and
we have an interesting regime. In this case, we evaluate
$R_{min}(\alpha,\beta)$ as follows.

\begin{theorem}\label{thm:smallbeta} When $\alpha$ is fixed, for any $\epsilon>0$, we
have
\begin{equation}
R_{min}(\alpha, \beta) = -\beta \log \beta + \beta
\left(\frac{1+h_2(\alpha)}{\alpha} + \log e - C \right)+
O(\beta^{2-\epsilon}),\label{eqn:smallbeta}
\end{equation}
where $C = \sum_{l=1}^{\infty}2^{-l-1}l \log l\approx 1.29$.
\end{theorem}

The proof of Theorem~\ref{thm:smallbeta} based on
Lemmas~\ref{lem:Rit} and \ref{lem:breakdown}, and is provided in
Appendix~\ref{app:proofsmallbeta}. Detailed discussions about the
proof techniques are given in
Section~\ref{subsec:fixedalphasmallbeta}.

\begin{remark}
The dominating term on the right side of (\ref{eqn:smallbeta}) is
$-\beta \log \beta$, and the second leading term is of order
$\Theta(\beta)$. Since $-\log \beta$ tends to infinity slowly as
$\beta$ decreases to zero, in practice these two terms are often in
the same order of magnitude. Therefore we need to evaluate both of
them.
\end{remark}

\begin{remark}\label{rem:Montanari}
In \cite{MontanariISIT10}, the authors evaluated the mutual
information across the iid deletion channel with iid
Bernoulli$(1/2)$ input as
\[\lim_{n\rightarrow \infty}\frac{1}{n}I(X^n;y(X^n,D^n))=1+d \log d - d (\log 2e -C)+O(d^{2-\epsilon}),\]
which implies that
\[\lim_{n\rightarrow \infty}\frac{1}{n}H(y(X^n,D^n)|X^n)=-d \log d + d (\log 2e -C)+O(d^{2-\epsilon}).\]
This expression should be compared with (\ref{eqn:smallbeta}) in the
special case that the deletion process is iid, which requires $\beta
= 1 - \alpha = d$. Under this condition, (\ref{eqn:smallbeta}) also
has the same two leading terms $-d \log d + d (\log 2e -C)$.
Therefore in the special case of iid deletion process,
(\ref{eqn:smallbeta}) is consistent with the result in
\cite{MontanariISIT10}.
\end{remark}

\begin{remark}\label{rem:channel}
Theorem~\ref{thm:smallbeta} implies that when the input distribution
is iid Bernoulli$(1/2)$, the mutual information across the bursty
deletion channel is
\begin{eqnarray}\label{eqn:channel}
\lefteqn{\lim_{n\rightarrow \infty}\frac{1}{n}I(X^n;y(X^n,D^n))}\nonumber\\
\!\!\!&=&\!\!\! 1+\beta \log \beta - \beta
\left(\frac{1+h_2(\alpha)}{\alpha} + \log e - C \right)+
O(\beta^{2-\epsilon}).
\end{eqnarray}
In \cite{Dobrushin}, Dobrushin showed that the channel capacity of
the iid deletion channel is $\lim_{n\rightarrow \infty} (1/n) \max
_{p_{X^n}} I(X^n; y(X^n,D^n))$. If this expression can be extended
to the bursty deletion channel where the deletion pattern process is
a Markov chain, then (\ref{eqn:channel}) provides an asymptotic
lower bound for the capacity of the bursty deletion channel for
small values of $\beta$.
\end{remark}

\section{Information Theoretic Expression for General $\alpha$ and $\beta$}
\label{sec:general} We can write the minimum achievable rate
$R_{min}$ as the following information theoretic expression.
\begin{lemma}\label{lem:Rit}
\[R_{min}=\lim_{n\rightarrow
\infty}\frac{1}{n}H(X^n|y(X^n,D^n),D_0,D_{n+1}).\]
\end{lemma}
The proof of Lemma~\ref{lem:Rit} is given in
Appendix~\ref{app:proofthmRit}. The structure of the proof is as
follows: (1) we show that the limit $\lim_{n\rightarrow
\infty}(1/n)H(X^n|y(X^n,D^n),D_0,D_{n+1})$ exists, (2) using the
information-spectrum method \cite[Section 7.2]{Han}, we have
$R_{min}= \overline{H}(X^n|y(X^n,D^n)):=$
$\mbox{p-}~\!\!\!\!\!\!\!\limsup_{n\rightarrow \infty} (1/n)\log
(1/p_{X^n|y(X^n,D^n)}(X^n|y(X^n,D^n)))$, which is the conditional
spectral sup-entropy, (3) we show that
$\overline{H}(X^n|y(X^n,D^n))=\lim_{n\rightarrow
\infty}(1/n)H(X^n|y(X^n,D^n),D_0,D_{n+1})$. The techniques we use in
step (3) are similar to those Dobrushin used in \cite{Dobrushin},
where the capacity of the iid deletion channel is characterized by
$\lim_{n\rightarrow \infty} (1/n) \max _{p_{X^n}} I(X^n;
y(X^n,D^n))$.



%

In Lemma~\ref{lem:breakdown}, the information theoretic expression
of the minimum rate is written in another way, which has a more
intuitive interpretation as explained in Remark~\ref{rem:breakdown}.

\begin{lemma}\label{lem:breakdown}
\begin{equation}
R_{min} = d + H(D_1|D_0) - E_{\infty},\label{eqn:breakdown}
\end{equation}
where $E_{\infty} := \lim_{n\rightarrow \infty} E_n$, and $E_n:=
H(D_1|D_0,X^n,y(X^n,D^n),D_{n+1})$.
\end{lemma}

The proof of Lemma~\ref{lem:breakdown} is given in
Appendix~\ref{app:proofbreakdown}.

\begin{remark}
\label{rem:breakdown} Lemma~\ref{lem:breakdown} expresses $R_{min}$
in terms of three parts, which can be intuitively interpreted as
follows. The first term $d$ is the fraction of deleted bits in
$X^n$. It represents the amount of information per source bit in the
deleted content, and thus the rate needed to send the deleted
content. The second term is the entropy rate of the deletion pattern
process, which is the rate needed to describe the locations of
deletions. If the encoder knew the locations and sent them together
with the deleted content, the decoder could reproduce $X^n$.
However, this is excessive information.
 In fact, even if the decoder can correctly reproduce $X^n$,
it can never know the exact deletion pattern. Therefore the
uncertainty of the deletion pattern $D^n$, given $X^n$ and
$y(X^n,D^n)$, is not required to be revealed in order to reproduce
$X^n$.

The uncertainty in the deletion pattern, given the source sequence
and side-information is the \emph{nature's secret}, which is known
only to an imaginary third party (nature) who generates the deletion
pattern. Since nature's secret is not required to reproduce $X^n$,
it should be subtracted from the message rate.
Lemma~\ref{lem:breakdown} shows that nature's secret per source bit,
which is the uncertainty in the whole deletion pattern $D^n$
normalized by $n$, can be expressed as $E_{\infty}$, which is the
uncertainty in only $D_1$. An intuitive explanation is that, the
uncertainty in each bit in $D^n$ is approximately the same,
therefore the uncertainty can be represented by the uncertainty in
only $D_1$.
\end{remark}

\section{Asymptotic behavior of $R_{min}$ for small values of $\beta$}
\label{sec:smallbeta}

In typical settings the number of edits is often much less than the
file size. Since $\beta$ is the probability to start a burst of
deletions, the asymptotic behavior of $R_{min}$ for small $\beta$ is
of special interest.

\subsection{Case 1: Few number of long bursts of deletion: $\alpha\ll 1, \beta\ll 1$, and $\alpha/\beta$ is fixed}
\label{subsec:smallalphasmallbeta}

When $\alpha\ll 1, \beta\ll 1$ and $\alpha/\beta$ is fixed, the
number of bursts are much smaller than the length of the sequence,
and each burst is so long that the overall fraction of deletion
$d=\beta/(\alpha+\beta)$ is a constant.

On the right side of (\ref{eqn:breakdown}), the first term $d$ is a
constant. For any $\epsilon>0$, the second term $H(D_1|D_0)=d
h_2(\alpha)+(1-d)h_2(\beta)= O(\beta^{1-\epsilon})$, and the third
term $E_{\infty}\leq H(D_1|D_0)=O(\beta^{1-\epsilon})$. According to
Lemma~\ref{lem:breakdown}, we have
\[R_{min}(\alpha, \beta) = d + O(\beta^{1-\epsilon}).\]
Intuitively speaking, if we have a small number of long bursts of
deletion, the amount of information of the locations of deletions is
orderwise less than the amount of information of the content of
deletion. Therefore $R_{min}$ is dominated by the rate needed to
deliver the deleted content.

A more interesting case is when all three terms of
(\ref{eqn:breakdown}) are comparable.

\subsection{Case 2: Few number of short bursts of deletion: $\alpha$ is fixed and $\beta\ll 1$}
\label{subsec:fixedalphasmallbeta}

When $\alpha$ is fixed and $\beta\ll 1$, the number of bursts is
much smaller than the length of the sequence. Since the length of a
burst is drawn from a geometric distribution with parameter
$\alpha$, the expected length is of order $\Theta(1)$. The overall
proportion of deleted bits is $d = \beta/(\alpha+\beta)
=\beta/\alpha + \Theta(\beta^2)$. In this case, unlike in Case~1,
the location information and ``nature's secret'' are comparable to
the content information. Therefore we need to evaluate all three
terms for this case. The three terms on the right side of
(\ref{eqn:breakdown}) are evaluated as follows. For any $\epsilon
>0$, we have
\begin{eqnarray}
d&=& \beta/\alpha + \Theta(\beta^2),\label{eqn:smallbetad}\\
H(D_1|D_0)&=& \!\!-\beta \log \beta + \frac{\beta
h_2(\alpha)}{\alpha} +
\beta \log e  + O(\beta^{2-\epsilon}),\label{eqn:smallbetaentropyrate}\\
-E_\infty &=& -C \beta +
O(\beta^{2-\epsilon}),\label{eqn:smallbetasecret}
\end{eqnarray}
where $C = \sum_{l=1}^{\infty}2^{-l-1}l \log l\approx 1.29$.
Combining (\ref{eqn:smallbetad}) through (\ref{eqn:smallbetasecret})
gives Theorem~\ref{thm:smallbeta}.

The proofs of (\ref{eqn:smallbetad}) and
(\ref{eqn:smallbetaentropyrate}) are trivial. The proof of
(\ref{eqn:smallbetasecret}) is highly nontrivial and is the essence
of the proof of Theorem~\ref{thm:smallbeta}. The complete proof of
(\ref{eqn:smallbetasecret}) is given in
Appendix~\ref{app:proofsmallbeta}. In this subsection we explain
only the intuition of (\ref{eqn:smallbetasecret}).


Let us first consider the case that the deletion is not bursty
($\alpha=1$), i.e., no consecutive bits are deleted. In order to
evaluate nature's secret $E_\infty$ we need to estimate the
uncertainty in $D_1$ given $X^n, y(X^n,D^n), D_0$ and $D_{n+1}$. The
uncertainty is significant if the first run of $X^n$ is different
from the first run of $y(X^n,D^n)$. For example, if $X^n=(0,0,0,1)$
and $y(X^n,D^n)=(0,0,1)$, we know that one bit is deleted in the
first run (first three bits) of $X^n$, but do not know which bit is
deleted. The true identity of the deleted bit is nature's secret.
Since there are three equally likely possible deletion patterns and
only one leads to $D_1=1$, the conditional entropy of $D_1$ is
$h_2(1/3)$. The length of the first run of $X^n$ is $L$, a
geometrically distributed random variable with parameter $1/2$. If
one bit is deleted in the first run, the conditional entropy is
$h_2(1/L)$. The probability that any bit in $L$ bits is deleted is
roughly $L\beta$, therefore the average uncertainty is
$\eE[h_2(1/L)L \beta]=\left(\sum_{l=1}^{\infty}h_2(1/l) 2^{-l} l
\right)\beta =\left(\sum_{l=1}^{\infty}2^{-l-1}l \log l\right)\beta
=C \beta$.\footnote{In this section we only provide an intuitive
explanation using a simplified case that there is only one burst of
deletion. In a rigorous proof it is shown that with high probability
the first burst of deletion can be isolated from the other bursts so
that the general case is reduced to the simplified case. See
Appendix~\ref{app:proofsmallbeta} for details.}

Let us now extend the discussion in the previous paragraph to the
case of bursty deletions ($\alpha <1$). First, we need to generalize
the usual definition of ``run'' to $b$-run.

\begin{definition}\label{def:brun}
For any $b$ and $l\in \zZ^+$, a sequence $(x_1,\ldots,x_{b+l-1})$ is
called a $b$-run of extent $l$ if for all $i,j$ satisfying $(i\equiv
j\mod b)$, $x_i=x_j$ holds.
\end{definition}

For example, $(1,1,1,1,1)$ is a $1$-run of extent $5$, and $1$-run
is the usual definition of a run. The sequence $(1,0,1,0,1)$ is a
$2$-run of extent $4$. Note that there are $l$ different ways to
delete $b$ consecutive bits in a sequence of length $l+b-1$. A
special property of a $b$-run of extent $l$ is that, all the $l$
ways of deletion result in the same outcome. For example, all four
ways of deleting two consecutive bits in $(1,0,1,0,1)$ lead to the
same outcome $(1,0,1)$. This observation is formally stated in the
following fact.
\begin{fact}\label{fact:deletionbrun}
Let $x^{b+l-1}$ be a $b$-run of extent $l$. Let $\mathbf d_{i,b}$
denote the sequence of $(i-1)$ $0$'s followed by $b$ $1$'s, then
followed by $(l-i)$ $0$'s. Then $y(x^{b+l-1}, \mathbf d_{i,b})$ is
the same for all $i=1,\ldots,l$.
\end{fact}


\begin{definition}\label{def:firstbrun}
For any $b\in \zZ^+$, the first $b$-run of a sequence
$(x_1,\ldots,x_{n})$ is the longest segment starting from $x_1$ that
is a $b$-run.
\end{definition}

For example, the first $2$-run of $(0,1,0,1,1)$ is $(0,1,0,1)$.

Now let us consider the uncertainty in $D_1$ given $X^n, y(X^n,D^n),
D_0$ and $D_{n+1}$ through an example. If we know that a burst of
$2$ bits is deleted in $X^n = (0,1,0,1,1)$ to produce
$y(X^n,D^n)=(0,1,1)$, we know that the deletion occurs within the
first $2$-run, i.e., $(0,1,0,1)$. Since there are three
indistinguishable deletion patterns, $(1,1,0,0,0)$, $(0,1,1,0,0)$,
and $(0,0,1,1,0)$, among which only the first one satisfies $D_1=1$,
the conditional entropy of $D_1$ is $h_2(1/3)$.

For any $b$, the extent of the first $b$-run, $L$, is a
geometrically distributed random variable with parameter $1/2$, as
in the non-bursty case. This fact can be seen by sequentially
generating $X_1, X_2, \ldots$. For arbitrary realization of
$X^b=x^b$,  $X^b$ always belongs to the first $b$-run. If the first
$b$-run has been extended to the $(i-1)$-th bit, it will be extended
to the $i$-th bit if $X_i=x_{i-b}$, which occurs with probability
$\frac{1}{2}$. Therefore the extent of the first $b$-run is a
geometrically distributed variable. If one burst of $b$ is deleted
in the first $b$-run, the conditional entropy of $D_1$ is
$h_2(1/L)$. Since given the length of burst $b$, the probability
that any deletion pattern among all $L$ possible deletion patterns
occurs is roughly $L\beta$, the average uncertainty of
$\eE[h_2(1/L)L\beta]=C\beta$. Note that the result is the same for
all $b$. In other words, nature's secret is always $C\approx 1.29$
bits \emph{per burst}, regardless of the length of burst.

\begin{remark}
Since nature's secret is $C\beta + O(\beta^{2-\epsilon})$ for
\emph{any given value of the length of burst} $b \in \zZ^+$, the
fact that nature's secret averaged across different possible values
of $b$ is $C\beta + O(\beta^{2-\epsilon})$, regardless of the
distribution of the length of a burst of deletions. This implies
that Theorem~\ref{thm:smallbeta} may generalize to more general
deletion processes beyond the two-state Markov chains.
In order to draw a rigorous statement, however, one has to revisit
Lemmas~\ref{lem:Rit} and \ref{lem:breakdown} and prove them for the
general setup.
\end{remark}

\section{Concluding Remarks}
We studied the distributed source coding problem of synchronizing
source sequences based on bursty deletion side-information. We
evaluated the two leading terms of the minimum achievable rate for
small deletion rate. Directions for future work include considering
insertions in addition to deletions, and evaluating the leading
terms of the capacity of the bursty deletion channel.


\appendices
\renewcommand{\theequation}{\thesection.\arabic{equation}}
\setcounter{equation}{0}

\section{Proof of Lemma~\ref{lem:Rit}}\label{app:proofthmRit}
(1) We first show that $R_n := (1/n)H(X^n|y(X^n,D^n),D_0,D_{n+1})$
converges as $n\rightarrow \infty$, so that the limit in the
statement of Lemma~\ref{lem:Rit} is well defined.

For all $m\in \{1,\ldots,n-1\}$, we have
\begin{eqnarray*}
nR_n&=&H(X^n|y(X^n,D^n),D_0,D_{n+1}) \\
&\stackrel{(a)}{\geq}&H(X^n|y(X^m,D^m),y(X_{m+1}^n,D_{m+1}^n),D_0,D_{n+1})\\
&\geq&H(X^m|y(X^m,D^m),y(X_{m+1}^n,D_{m+1}^n),D_0,D_{n+1},D_{m+1})\\
&&+H(X_{m+1}^n|y(X^m,D^m),y(X_{m+1}^n,D_{m+1}^n),D_0,D_{n+1},D_m)\\
&\stackrel{(b)}{=}&H(X^m|y(X^m,D^m),D_0,D_{m+1})\\
&& + H(X_{m+1}^n|y(X_{m+1}^n,D_{m+1}^n),D_{n+1},D_m)\\
&=&H(X^m|y(X^m,D^m),D_0,D_{m+1})\\
&& + H(X^{n-m}|y(X^{n-m},D^{n-m}),D_{0},D_{n-m})
\end{eqnarray*}
where step (a) holds because the tuple
$(y(X^m,D^m),y(X_{m+1}^n,D_{m+1}^n))$ determines $y(X^n,D^n)$, and
step (b) holds because the Markov chains
$(y(X_{m+1}^n,D_{m+1}^n),D_{n+1})-D_{m+1}-(X^m,y(X^m,D^m),D_0)$ and
$ (y(X^m,D^m),D_0) - D_m -
(X_{m+1}^n,y(X_{m+1}^n,D_{m+1}^n),D_{n+1})$ hold. Therefore the
sequence $\{n R_n \}_{n\in \nN}$ is superadditive. By Fekete's
lemma\cite{Feketeslemma}, the limit $\lim_{n\rightarrow \infty} R_n$
exists.

(2) Using the information-spectral version of the Slepian-Wolf
theorem \cite[Section 7.2]{Han}, we have $R_{min}=
\overline{H}(X^n|y(X^n,D^n)):=$
$\mbox{p-}~\!\!\!\limsup_{n\rightarrow \infty} (1/n)\log
(1/p_{X^n|y(X^n,D^n)}(X^n|y(X^n,D^n)))$. In the rest of this
appendix, for any random variables $A, B$, we abbreviate $p_A(A)$
and $p_{A|B}(A|B)$ to $p(A)$ and $p(A|B)$, respectively, to avoid
cumbersome notations.

(3) Now we show that the sequence of random variables $(1/n)\log
(1/p(X^n|y(X^n,D^n)))$ converges in probability to the limit
$\lim_{n\rightarrow \infty} R_n$.

We introduce a segmented deletion process as follows. Let $k \geq 3$
be the length of a segment. Let $g :=\lfloor n/k \rfloor$ be the
number of complete segments and $l := n-gk$ be the length of the
remainder. Consider the outcome of a segmented deletion process as
follows: let $z(X^n, D^n) := (Z_{1L},Z_{1M},Z_{1R},
\ldots,Z_{gL},Z_{gM},Z_{gR}, Z_{remainder})$ be a vector with
$(3g+1)$ components, where $\forall i= 1,\ldots,g$, $Z_{iL}:=
y(X_{(i-1)k+1},D_{(i-1)k+1})$, $Z_{iM}:=
y(X_{(i-1)k+2}^{ik-1},D_{(i-1)k+1}^{ik-1})$, $Z_{iR}:= y(X_{i
k},D_{i k})$, and $Z_{remainder} := y(X_{gk+1}^{n},D_{gk+1}^{n})$.
From $z(X^n, D^n)$ we can find out how many source bits are deleted
in each segment and the remainder, and whether the first and last
bits of each segment are deleted. The sequence $y(X^n, D^n)$ can be
obtained by merging all the $(3g+1)$ components of $z(X^n, D^n)$.
Therefore the sequence $z(X^n, D^n)$ contains more information than
$y(X^n, D^n)$. We will first fix $k$ and let $n$ go to infinity.
Then we increase $k$ to prove the final result.

The statement to be proved is based on the following three facts.

\begin{fact}\label{fact:YtoZ} For any $k\geq 3$, $n$ and any $\delta>0$, there
exists a function $\epsilon_1(k)$ satisfying $\lim_{k\rightarrow
\infty} \epsilon_1(k)=0$, so that
\begin{equation*}
\pP\left(\frac{1}{n}\left|\log \frac{1}{p(X^n|y(X^n,D^n))} - \log
\frac{1}{p(X^n|z(X^n,D^n))}\right| > \delta \right)\leq
\frac{\epsilon_1(k)}{\delta}.
\end{equation*}
\end{fact}

\begin{fact}\label{fact:ZtoH} For any $k$ and any $\delta>0$, there
exists a function $\epsilon_2(k)$ satisfying $\lim_{k\rightarrow
\infty} \epsilon_2(k)=0$, so that as $n\rightarrow \infty$,
\begin{eqnarray*}
\lefteqn{\pP\left(\left|\frac{1}{n} \log
\frac{1}{p(X^n|z(X^n,D^n))}\right.\right.}\\
&&\left.\left.-\frac{1}{k}
H(X_2^{k-1}|y(X_2^{k-1},D_2^{k-1}),D_1,D_k)\right|>\delta\right)\\
&\leq& \frac{\epsilon_2(k)}{\delta}.
\end{eqnarray*}
\end{fact}

\begin{fact}\label{fact:HtoRmin}
\begin{eqnarray*}
\lefteqn{\lim_{k\rightarrow\infty} \frac{1}{k}
H(X_2^{k-1}|y(X_2^{k-1},D_2^{k-1}),D_1,D_k)}\\
&=&\lim_{k\rightarrow\infty} \frac{1}{k}
H(X^{k}|y(X^{k},D^{k}),D_0,D_{k+1}).
\end{eqnarray*}
\end{fact}

\noindent {\em Proof of Fact~\ref{fact:YtoZ}:}

Since $y(X^n,D^n)$ can be determined by $z(X^n,D^n)$, there exists a
function $\phi_n$ such that $y(X^n,D^n) = \phi_n(z(X^n,D^n))$. For
any realization of $z(X^n,D^n)=z$, we have $\pP(z(X^n,D^n)=z)\leq
\pP(y(X^n,D^n)=\phi_n(z))$, which implies that $(1/n)\log
\pP(z(X^n,D^n)=z)- (1/n)\log\pP(y(X^n,D^n)=\phi_n(z))\leq 0$ always
holds. Let $L_Z$ be the vector of (3g+1) components representing the
lengths of all the components of $z(X^n,D^n)$. Then we have

\begin{eqnarray*}
\lefteqn{\eE \left| \frac{1}{n}\log p(y(X^n,D^n)) -
\frac{1}{n}\log p(z(X^n,D^n))\right|}\\
&=& \eE\left[\frac{1}{n}\log p(y(X^n,D^n))\right] -
\eE\left[\frac{1}{n}\log p(z(X^n,D^n))\right]\\
&=& \frac{1}{n}(-H(y(X^n,D^n)) + H(z(X^n,D^n)))\\
&=& \frac{1}{n}H(z(X^n,D^n)|y(X^n,D^n))\\
&=& \frac{1}{n}H(L_Z|y(X^n,D^n))\\
&\leq& \frac{1}{n}H(L_Z)\\
&\leq& \frac{1}{n}(3g + 1)\log k\\
&\leq& \frac{4 \log k}{k}.
\end{eqnarray*}
By Markov's inequality,
\[\pP\left(\frac{1}{n}\left|\log p(y(X^n,D^n)) -
\log p(z(X^n,D^n))\right|>\delta\right) \leq \frac{4 \log k}{k
\delta}.\]

Using the same argument we also have
\[\pP\left(\frac{1}{n}\left|\log p(X^n,y(X^n,D^n)) -
\log p(X^n,z(X^n,D^n))\right|>\delta\right) \leq \frac{4 \log k}{k
\delta}.\] Combining the last two inequalities completes the proof
of Fact~\ref{fact:YtoZ}. \hspace{\fill}\QED
\\

\noindent {\em Proof of Fact~\ref{fact:ZtoH}:}

Let $Z_B:=(Z_{1L},Z_{1R},\ldots,Z_{gL},Z_{gR},Z_{remainder})$. Then
\begin{eqnarray}
\lefteqn{\frac{1}{n}\log p(z(X^n,D^n))}\nonumber\\
&\stackrel{(c)}{=}& \frac{1}{n}\log p(Z_B) + \sum_{i=1}^g
\frac{1}{n}\log p(Z_{iM}|Z_B)\nonumber\\
&\stackrel{(d)}{=}& \frac{1}{n}\log p(Z_B) + \sum_{i=1}^g
\frac{1}{n}\log p(Z_{iM}|Z_{iL},Z_{iR}),\label{eqn:lln}
\end{eqnarray}
where step (c) holds because given $Z_B$, $Z_{1M},\ldots,Z_{gM}$ are
conditionally independent, and step (d) holds because $D^n$ is a
Markov chain.

Since the expectation of the first term of (\ref{eqn:lln}) is equal
to $(1/n)H(Z_B)\leq (2g+l)/n \log 3$, by Markov's inequality we have
$\pP((1/n)\log p(Z_B)>\delta)<(2g+l)\log 3/(n\delta)$.

Due to the law of large number, as $n\rightarrow \infty$, which
implies $g\rightarrow \infty$, the second term of (\ref{eqn:lln})
converges to $(1/k)H(y(X_{2}^{k-1},D_{2}^{k-1})|D_1,D_k)$ in
probability.

Therefore we have: for any $k$ and $n\rightarrow \infty$,
\begin{equation*}
\pP\left(\left|\frac{1}{n} \log \frac{1}{p(z(X^n,D^n))}-\frac{1}{k}
H(y(X_2^{k-1},D_2^{k-1}),D_1,D_k)\right|>\delta\right) \leq
\frac{\epsilon'_2(k)}{\delta}
\end{equation*}
for some $\epsilon'_2(k)$ which vanishes as $k$ increases.

Using the same argument we also have
\begin{eqnarray*}
\lefteqn{\pP\left(\left|\frac{1}{n} \log
\frac{1}{p(X^n,z(X^n,D^n))}\right.\right.}\\
&&\left.\left.-\frac{1}{k}
H(X_2^{k-1},y(X_2^{k-1},D_2^{k-1}),D_1,D_k)\right|>\delta\right)\\
&\leq& \frac{\epsilon''_2(k)}{\delta}.
\end{eqnarray*}
Combining the last two inequalities completes the proof of
Fact~\ref{fact:ZtoH} \hspace{\fill}\QED
\\

\noindent {\em Proof of Fact~\ref{fact:HtoRmin}:}
Fact~\ref{fact:HtoRmin} holds because (i) $p_{X_2^{k-1},D_2^{k-1}} =
p_{X^{k-2},D^{k-2}}$ and (ii) $(k-2)/k \rightarrow 1$ as
$k\rightarrow \infty$.
 \hspace{\fill}\QED

Combining Facts~\ref{fact:YtoZ} and \ref{fact:ZtoH}, we have: for
any fixed $k$ and $\delta$,  as $n\rightarrow \infty$,
\begin{eqnarray}
\lefteqn{\pP\left(\left|\frac{1}{n} \log
\frac{1}{p(X^n|y(X^n,D^n))}\right.\right.}\nonumber\\
&&\left.\left.-\frac{1}{k}
H(X_2^{k-1}|y(X_2^{k-1},D_2^{k-1}),D_1,D_k)\right|>\delta\right)\nonumber\\
&\leq& \frac{\epsilon_3(k)}{\delta}\label{eqn:spectrum2}
\end{eqnarray}
for some $\epsilon_3(k)$ which vanishes as $k$ increases. By
choosing a large enough $k$, the right hand side of
(\ref{eqn:spectrum2}) can be made arbitrarily small. Combining
(\ref{eqn:spectrum2}) and Fact~\ref{fact:HtoRmin}, the sequence of
random variables $(1/n)\log (1/p(X^n|y(X^n,D^n)))$ is shown to be
converging in probability to the limit $\lim_{n\rightarrow \infty}
R_n$.

 Combining (1), (2) and (3) we have $R_{min} =
\lim_{n\rightarrow \infty} (1/n)H(X^n|y(X^n,D^n),D_0,D_{n+1})$.

\section{Proof of Lemma~\ref{lem:breakdown}}\label{app:proofbreakdown}

We will first introduce a sequence $\{J_n\}_{n\in \nN}$ and show
that $\lim_{n\rightarrow \infty} J_n = R_{min}$.

\begin{lemma}\label{lem:Jn}
For all $n\in \zZ^+$,  let $J_n:= d +
(1/n)H(y(X^n,D^n)|X^n,D_0,D_{n+1})$. Then we have
$\lim_{n\rightarrow \infty} J_n = R_{min}$.
\end{lemma}

\begin{proof}
We have
\begin{eqnarray*}
R_{min}&=&\lim_{n\rightarrow \infty} \frac{1}{n}
H(X^n|y(X^n,D^n),D_0,D_{n+1})\\
&=&\lim_{n\rightarrow \infty}
\frac{1}{n}\big[H(X^n|D_0,D_{n+1})+H(y(X^n,D^n)|X^n,D_0,D_{n+1})\\
&&\ \ \ \ \ \ \ \ \ -H(y(X^n,D^n)|D_0,D_{n+1})\big]\\
&=& 1 + \lim_{n\rightarrow \infty}
\frac{1}{n}H(y(X^n,D^n)|X^n,D_0,D_{n+1})\\
&&-\lim_{n\rightarrow \infty}\frac{1}{n}( H(L_y|D_0,D_{n+1}) +
H(y(X^n,D^n)|L_y,D_0,D_{n+1})).
\end{eqnarray*}

Since
\[0\leq \lim_{n\rightarrow \infty}\frac{1}{n}H(L_y|D_0,D_{n+1}) \leq
\lim_{n\rightarrow \infty}\frac{1}{n}\log(n+1)=0,\] we have
$\lim_{n\rightarrow \infty}\frac{1}{n}H(L_y|D_0,D_{n+1})=0$. Since
given $L_y=l$ and given $(D_0,D_{n+1})$ the sequence $y(X^n,D^n)$ is
an iid Bernoulli$(1/2)$ sequence,
$H(y(X^n,D^n)|L_y=l,D_0,D_{n+t})=l$ holds. Therefore
$H(y(X^n,D^n)|L_y,D_0,D_{n+t})=\eE(L_y)$ and hence
\[\lim_{n\rightarrow \infty}\frac{1}{n}
H(y(X^n,D^n)|L_y,D_0,D_{n+1})=\lim_{n\rightarrow \infty}
\frac{1}{n}\eE[L_y]=\frac{\alpha}{\alpha+\beta}=1-d.\]

In conclusion,
\begin{eqnarray*}
R_{min}&=& 1 + \lim_{n\rightarrow \infty}
\frac{1}{n}H(y(X^n,D^n)|X^n,D_0,D_{n+1})-(1-d)\\
&=& \lim_{n\rightarrow \infty}\left(d +
\frac{1}{n}H(y(X^n,D^n|X^n,D_0,D_{n+1}))\right)\\
&=&\lim_{n\rightarrow \infty} J_n,
\end{eqnarray*}
which completes the proof of Lemma~\ref{lem:Jn}.
\end{proof}

Now let us use Lemma~\ref{lem:Jn} to prove
Lemma~\ref{lem:breakdown}.

Expanding $I(D_1;y(X^n,D^n)|X^n,D_0,D_{n+1})$ in two ways, we have
\begin{eqnarray}
\lefteqn{H(D_1|X^n,D_0,D_{n+1})-H(D_1|X^n,y(X^n,D^n),D_0,D_{n+1})}\nonumber\\
&=&
H(y(X^n,D^n)|X^n,D_0,D_{n+1})-H(y(X^n,D^n)|X^n,D_0,D_1,D_{n+1}).\nonumber\\
\label{eqn:fourterms}
\end{eqnarray}
The first term on the left side of (\ref{eqn:fourterms}) is equal to
$H(D_1|D_0,D_{n+1})$. The second term on the left side of
(\ref{eqn:fourterms}) is denoted by $E_n$. The first term on the
right side of (\ref{eqn:fourterms}) is equal to $n(J_n - d)$. The
second term on the right side of (\ref{eqn:fourterms}) is:
\begin{eqnarray*}
\lefteqn{H(y(X^n,D^n)|X^n,D_0,D_1,D_{n+1})}\\
&=&H(y(X^n,D^n)|X^n,D_1,D_{n+1})\\
&=&H(y(X^n,D^n)|X^n,D_1=1,D_{n+1})p_{D_1}(1)\\
&&+H(y(X^n,D^n)|X^n,D_1=0,D_{n+1})p_{D_1}(0)\\
&=&H(y(X_2^n,D_2^n)|X_1,X_2^n,D_1=1,D_{n+1})p_{D_1}(1)\\
&&+H(X_1,y(X_2^n,D_2^n)|X_1,X_2^n,D_1=0,D_{n+1})p_{D_1}(0)\\
&\stackrel{(e)}{=}&H(y(X_2^n,D_2^n)|X_2^n,D_1=1,D_{n+1})p_{D_1}(1)\\
&&+H(y(X_2^n,D_2^n)|X_2^n,D_1=0,D_{n+1})p_{D_1}(0)\\
&=&H(y(X_2^n,D_2^n)|X_2^n,D_1,D_{n+1})\\
&=&H(y(X^{n-1},D^{n-1})|X^{n-1},D_0,D_{n})\\
&=&(n-1)(J_{n-1}-d),
\end{eqnarray*}
where step (e) holds because $X_1$ is independent of
$(D^{n+1},X_2^n, y(X_2^n,D_2^n))$. Therefore (\ref{eqn:fourterms})
becomes
\begin{equation} H(D_1|D_0,D_{n+1}) -
E_n = n(J_n - J_{n-1}) +J_{n-1} -
d.\label{eqn:newfourterms}\end{equation}

Now let us take the limit as $n\rightarrow \infty$ on both sides of
(\ref{eqn:newfourterms}). Because of mixing of the Markov chain
$\{D_i\}_{i\geq 0}$, the distribution
$p_{D_{n+1}|D_0,D_1}(\cdot|d_0,d_1)$ converges to the stationary
distribution regardless of the initial values $(d_0,d_1)$ as $n$
goes to infinity. Therefore $\lim_{n\rightarrow \infty}
H(D_1|D_0,D_{n+1})=H(D_1|D_0)$. For the second term on the left side
of (\ref{eqn:newfourterms}), Lemma~\ref{lem:Enincreasing} guarantees
the convergence of $\{E_n\}_{n\geq 1}$.

\begin{lemma}\label{lem:Enincreasing}
(1) The sequence $\{E_n\}_{n\geq 1}$ is nondecreasing. (2)
$\lim_{n\rightarrow} E_n$ exists.
\end{lemma}
\begin{proof}
(1) For all $n\geq 2$, we have
\begin{eqnarray*}
E_n &=& H(D_1|X^n,y(X^n,D^n),D_0,D_{n+1})\\
&\geq & H(D_1|X^n,y(X^n,D^n),D_0,D_n,D_{n+1})\\
&=& H(D_1|X^n,y(X^n,D^n),D_0,D_n)\\
&=& H(D_1|X^n,y(X^n,D^n),D_0,D_n=1)p_{D_n}(1)\\
&&+H(D_1|X^n,y(X^n,D^n),D_0,D_n=0)p_{D_n}(0)\\
&=& H(D_1|X^{n-1},X_n,y(X^{n-1},D^{n-1}),D_0,D_n=1)p_{D_n}(1)\\
&&+H(D_1|X^{n-1},X_n,y(X^{n-1},D^{n-1}),D_0,D_n=0)p_{D_n}(0)\\
&=& H(D_1|X^{n-1},y(X^{n-1},D^{n-1}),D_0,D_n)\\
&=& E_{n-1}.
\end{eqnarray*}
Therefore $\{E_n\}_{n\geq 1}$ is nondecreasing.

(2) Since for all $n$, $E_n \geq 1$ holds and $\{E_n\}_{n\geq 1}$ is
nondecreasing, $E_{\infty}=\lim_{n\rightarrow} E_n$ exists.
\end{proof}

By Lemma~\ref{lem:Enincreasing}, the left side of
(\ref{eqn:newfourterms}) converges to $H(D_1|D_0)-E_{\infty}$ as
$n\rightarrow \infty$. Since (\ref{eqn:newfourterms}) holds, the
right side also converges and the limit is $\left(\lim_{n\rightarrow
\infty} n(J_n - J_{n-1})\right) + R_{min}-d$. Since $\{J_n\}_{n\geq
1}$ is a converging sequence and the $\lim_{n\rightarrow \infty}
n(J_n - J_{n-1})$ exists, $\lim_{n\rightarrow \infty} n(J_n -
J_{n-1})=0$. Therefore in the limit as $n\rightarrow \infty$,
(\ref{eqn:newfourterms}) becomes
\[H(D_1|D_0)-E_{\infty} = R_{min} - d, \]
which completes the proof of Lemma~\ref{lem:breakdown}.

\section{Proof of Theorem~\ref{thm:smallbeta}}\label{app:proofsmallbeta}

 When $\alpha$ is a fixed constant and $\beta\ll 1$, it is easy
to verify that the first two terms of (\ref{eqn:breakdown}) are
\begin{eqnarray*}
d+H(D_1|D_0)&=& \frac{\beta}{\alpha+\beta} + \frac{\alpha
h_2(\beta)}{\alpha+\beta} + \frac{\beta
h_2(\alpha)}{\alpha+\beta}\\
&=&-\beta \log \beta + \beta \left(\frac{1+h_2(\alpha)}{\alpha} +
\log e\right)+ O(\beta^{2-\epsilon}),
\end{eqnarray*}
for any $\epsilon >0$. We will show that the third term of
(\ref{eqn:breakdown}) $E_\infty = C \beta + O(\beta^{2-\epsilon})$.

Let us first define ``typicality'' of the deletion pattern. Since
$E_\infty$ is the conditional entropy of $D_1$, which is more
relevant to the first a few bits of $D_0^n$, the typicality of the
$D_0^n$ concerns about only the first a few bits.
\begin{definition}\label{def:typicality}
Let $k = \max \{6, 6/(\log (1-\alpha))\}$. For $n > -k \log \beta$,
the deletion pattern $D_0^n$ is typical if the following two
conditions hold.
\begin{enumerate}
  \item There is at most one run of $1$'s in $(D_0,\ldots,D_{-k \log
  \beta})$.
  \item There are no more than $(-k/3 \log \beta)$ $1$'s in $(D_0,\ldots,D_{-k \log
  \beta})$.
\end{enumerate}
\end{definition}

Lemma~\ref{lem:typicality} states that the deletion pattern is
typical with high probability.
\begin{lemma}\label{lem:typicality}
For  any $\epsilon>0$, the probability that $D_0^n$ is typical is at
least $1-O(\beta^{2-\epsilon})$.
\end{lemma}
\begin{proof}
Since any deletion pattern that has $r$ runs of $1$'s in
$(D_0,\ldots,D_{-k \log \beta})$ occurs with probability
$O(\beta^r)$ and there are no more than $(-k \log
  \beta)^{2r}$ such patterns, $\pP((D_0,\ldots,D_{-k \log \beta})$
  contains $r$ runs of $1$'s$)=O(\beta^{r-\epsilon})$ for any
  $\epsilon>0$. Hence condition 1) of
  Definition~\ref{def:typicality} holds with probability
  $1-O(\beta^{2-\epsilon})$. Given that condition 1) holds,
condition 2) is violated if there is a burst of deletion longer than
$(-k/3 \log \beta)$, which occurs with the probability
$O((1-\alpha)^{-k/3 \log \beta})=O(\beta^2)$. In conclusion,
$\pP(D_0^n$ is typical $)=1-O(\beta^{2-\epsilon})$ for any
$\epsilon>0$.
\end{proof}

Let the indicator random variable $T:=1$ if $D_0^n$ is typical and
$T:=0$ otherwise. Lemma~\ref{lem:typicality} implies that
$p_T(0)=O(\beta^{2-\epsilon}), \forall \epsilon>0$.
Lemma~\ref{lem:focusontypical} states that we can focus on the
typical case $T=1$ in order to evaluate $E_\infty$ to the precision
of $O(\beta^{2-\epsilon})$.
\begin{lemma}\label{lem:focusontypical}
\[E_{\infty} = \lim_{n\rightarrow
\infty}
H(D_1|X^n,y(X^n,D^n),D_0,D_{n+1},T=1)p_T(1)+O(\beta^{2-\epsilon}).\]
\end{lemma}
\begin{proof}
For all $n > -k \log \beta$, we have the following lower bound of
$E_n$\begin{eqnarray*}
E_n &\geq& H(D_1|X^n,y(X^n,D^n),D_0,D_{n+1},T)\nonumber\\
&\geq& H(D_1|X^n,y(X^n,D^n),D_0,D_{n+1},T=1)p_T(1),
\end{eqnarray*}
and the following upper bound
\begin{eqnarray*}
E_n &\leq& H(D_1,T|X^n,y(X^n,D^n),D_0,D_{n+1})\nonumber\\
&=& H(D_1|X^n,y(X^n,D^n),D_0,D_{n+1},T)\nonumber\\
&&+ H(T|X^n,y(X^n,D^n),D_0,D_{n+1})\nonumber\\
&\leq& H(D_1|X^n,y(X^n,D^n),D_0,D_{n+1},T=1)p_T(1)\nonumber\\
&&+ H(D_1|X^n,y(X^n,D^n),D_0,D_{n+1},T=0)p_T(0)+H(T) \nonumber\\
&\leq& H(D_1|X^n,y(X^n,D^n),D_0,D_{n+1},T=1)p_T(1)\nonumber\\
&& +p_T(0)+H(T)\nonumber\\
&=&
H(D_1|X^n,y(X^n,D^n),D_0,D_{n+1},T=1)p_T(1)+O(\beta^{2-\epsilon}).
\end{eqnarray*}
Taking the limit as $n\rightarrow \infty$ completes the proof.
\end{proof}

For all $n > -k \log \beta$, we have
\begin{eqnarray*}
\lefteqn{H(D_1|X^n,y(X^n,D^n),D_0,D_{n+1},T=1)p_T(1)}\\
&=&H(D_1|X^n,y(X^n,D^n),D_0=1,D_{n+1},T=1)p_{D_0,T}(1,1)\\
&&+H(D_1|X^n,y(X^n,D^n),D_0=0,D_{n+1},T=1)p_{D_0,T}(0,1).
\end{eqnarray*}
We will separately analyze the following two cases: (1) $D_0=1, T=1$
and (2) $D_0=0,T=1$.

\begin{itemize}
\item Case (1): $D_0=1, T=1$. In this case we check
whether $X^{-k\log\beta} =  Y^{-k\log\beta}$. Let $M_1 := 1$ if they
match and $M_1:=0$ otherwise. Note that $M_1$ is determined by $X^n$
and $y(X^n,D^n)$.

\begin{itemize}
\item Case (1.1): $D_0=1, T=1, M_1=0$. There exists at least one $1$ in $D_1^{-k\log
\beta}$. Since $D_0=1$ and there is at most one run of $1$ in
$D_0^{-k\log \beta}$ in a typical deletion pattern, $D_1=1$ must
hold. Therefore $H(D_1|D_0=1,T=1,M_1=0)=0$.

\item Case (1.2): $D_0=1, T=1, M_1=1$.
In this case, both $D_1=0$ and $D_1=1$ are possible. Given $D_0=1,
T=1$, if $D_1=0$, then for all $i=2,\ldots,-k\log\beta$, $D_i=0$,
which implies that $X^{-k\log\beta} =  Y^{-k\log\beta}$. If $D_1=1$,
then for all $i=1,\ldots,-k\log \beta$, $X_i$ and $Y_i$ are
independently generated fair bits, hence the event $X_i=Y_i$ occurs
with probability $1/2$. Since events $\{X_i=Y_i\}_i$ are independent
across $i$, $\pP(M_1=1|D_1=1,D_0=1,T=1)=(1/2)^{k\log
\beta}=O(\beta^6)$. Since $\pP(D_1=1|D_0=1,T=1)=\Theta(1)$ and
$\pP(D_1=0|D_0=1,T=1)=\Theta(1)$, by Bayes' rule, we have
$\pP(D_1=1|D_0=1,T=1,M_1=1)=O(\beta^6)$. Therefore
$H(D_1|D_0=1,T=1,M_1=1)=O(\beta^{6-\epsilon}), \forall \epsilon>0$.
\end{itemize}

In conclusion, the contribution of Case (1) to $E_\infty$ is
\begin{eqnarray*}
\lefteqn{H(D_1|X^n,y(X^n,D^n),(D_0,T)=(1,1),D_{n+1})p_{D_0,T}(1,1)}\\
&=&H(D_1|X^n,y(X^n,D^n),(D_0,T)=(1,1),D_{n+1},M_1)\\
&&\times p_{D_0,T}(1,1)\\
&=&O(\beta^{6-\epsilon}).
\end{eqnarray*}

\item Case (2): $D_0=0, T=1$. In this case we will first check whether
$X^{-k/3\log\beta} =  Y^{-k/3\log\beta}$. Let $M_2 := 1$ if they
match and $M_2:=0$ otherwise.

\begin{itemize}
\item Case (2.1): $D_0=0, T=1, M_2=1$.
By the same argument as in Case~1 for $M_1=1$, we have
$\pP(D_1=1|D_0=0,T=1,M_2=1)=O(\beta^2)$, and
$H(D_1|D_0=0,T=1,M_2=1)p_{D_0,T,M_2}(0,1,1)=O(\beta^{2-\epsilon}),
\forall \epsilon>0$.

\item Case (2.2): $D_0=0, T=1, M_2=0$.
We try to find a length-$(-k/3 \log \beta)$ segment in $Y^{-k \log
\beta}$ that matches $X_{-2k/3 \log \beta +1}^{-k \log \beta}$.
Since (i) $M_2=0$ implies that at least one bit in the first
${-k/3\log\beta}$ bits is deleted and (ii) a burst of deletion in a
typical deletion pattern is no longer than ${-k/3\log\beta}$, there
must be no deletion in $D_{-2k/3 \log \beta+1}^{-k \log \beta}$,
which implies that there must be at least one segment in $Y^{-k \log
\beta}$ that matches $X_{-2k/3 \log \beta+1}^{-k \log \beta}$.
Define $B := 0$ if there are two or more segments that match
$X_{-2k/3 \log \beta}^{-k \log \beta}$; and for $b\in \zZ^+$, define
$B :=b$ if there is a unique segment $Y_{-2k/3 \log \beta+1 - b}^{-k
\log \beta -b}$ that matches $X_{-2k/3 \log \beta+1}^{-k \log
\beta}$ with an offset $b$.

\begin{itemize}
\item Case (2.2.1): $D_0=0, T=1, M_2=0, B=0$.
The condition $B=0$ requires at least $(-k/3 \log \beta)$
independent bit-wise matches, each of which occurs with probability
$(1/2)$. Hence $B=0$ occurs with probability at most $(1/2)^{-k/3
\log \beta}=O(\beta^2)$. Therefore the contribution of Case (2.2.1)
is
$H(D_1|D_0=0,T=1,M_2=0,B=0)p_{D_0,T,M_2,B}(0,1,0,0)=O(\beta^{2})$.

\item Case (2.2.2): $D_0=0, T=1, M_2=0, B=b \in \zZ^+$.
There must be a burst of deletion of length $b$ taking place in
$D_1^{-2k/3 \log \beta}$ which causes the offset of $b$ between
$X_{-2k/3 \log \beta +1}^{-k \log \beta}$ and the matching segment
in $y(X^n,D^n)$. Since the length of the burst is bounded by
$(-k/3\log \beta)$ in a typical deletion pattern, $b\leq (-k/3\log
\beta)$ must hold. Since we can find a correct correspondence
between a segment of $X^n$ to its outcome of deletion, the deletion
process to the left of the segment is conditionally independent to
the deletion process to the right. Therefore in order to evaluate
the conditional entropy of $D_1$ we need to focus on the process to
the left of the segment only.
Hence the contribution of this case to $E_\infty$ is: $\sum_{b}
H(D_1|X^n, y(X^n,
D^n),D_{n+1},T=1,D_0=0,M_2=0,B=b)p_{T,D_0,M_2,B}(1,0,0,b) =
\sum_{b}H(D_1|X^{n'}, y(X^{n'},
D^{n'}),T=1,D_0=0,M_2=0,B=b)p_{T,D_0,M_2,B}(1,0,0,b)$, where
$n':={-2k/3 \log \beta}$. Lemma~\ref{lem:cbeta} will show that the
contribution of Case (2.2.2) is $C \beta + O(\beta^{2-\epsilon})$.
This is the only case that is responsible for the leading term $C
\beta$ in $E_\infty$.
\end{itemize}
\end{itemize}
\end{itemize}

As a summary, the contribution of all the cases (1.1), (1.2), (2.1),
(2.2.1) to $E_\infty$ is of order $O(\beta^{2-\epsilon})$.
Lemma~\ref{lem:cbeta} will show that the contribution of Case
(2.2.2) is $C \beta + O(\beta^{2-\epsilon})$, which will complete
the proof of Theorem~\ref{thm:smallbeta}.

\begin{lemma}\label{lem:cbeta} For $n':={-2k/3 \log \beta}$, we
have $\sum_{b=1}^{-k/3\log\beta} H(D_1|X^{n'}, y(X^{n'},
D^{n'}),(T,D_0,M_2,B)=(1,0,0,b) ) \times p_{T,D_0,M_2,B}(1,0,0,b) =
C \beta +O(\beta^{2-\epsilon})$.
\end{lemma}
\begin{proof}
Using the abbreviation $Y:=y(X^{n'},D^{n'})$, we have
\begin{eqnarray}
\lefteqn{\sum_{b=1}^{-k/3\log\beta} H(D_1|X^{n'}, Y,
(T,D_0,M_2,B)=(1,0,0,b))}\nonumber\\&&\times p_{T,D_0,M_2,B}(1,0,0,b)\nonumber\\
&=& \! \! \! \!\sum_{b=1}^{-k/3\log\beta}  \! \!\sum_{x^{n'},y}
H(D_1|X^{n'}=x^{n'}, Y=y,\nonumber\\
&&\hspace{0.5in}(T,D_0,M_2,B)=(1,0,0,b))\nonumber\\
&&\hspace{0.5in}\times
p_{X^{n'},Y,T,D_0,M_2,B}(x^{n'},y,1,0,0,b)\label{eqn:entropyexpansion}\\
&\stackrel{(f)}{=}& \! \! \! \!\sum_{b=1}^{-k/3\log\beta}  \!
\!\sum_{x^{n'}}
H(D_1|X^{n'}=x^{n'}, Y=x_{b+1}^{n'},\nonumber\\
&&\hspace{0.5in}(T,D_0,M_2,B)=(1,0,0,b))\nonumber\\
&&\hspace{0.5in}\times
p_{X^{n'},Y,T,D_0,M_2,B}(x^{n'},x_{b+1}^{n'},1,0,0,b)\nonumber\\
&\stackrel{(g)}{=}& \! \! \! \!\sum_{b=1}^{-k/3\log\beta}  \!
\!\sum_{x^{n'}}
H(D_1|X^{n'}=x^{n'}, Y=x_{b+1}^{n'},(T,D_0,B)=(1,0,b))\nonumber\\
&&\hspace{0.5in}\times
p_{X^{n'},Y,T,D_0,B}(x^{n'},x_{b+1}^{n'},1,0,b)+O(\beta^2),\label{eqn:entropyexpansion2}
\end{eqnarray}
where step (f) holds because of the following reason. Given
$(T,D_0,M_2,B)=(1,0,0,b)$, if $D_1=1$, then $D_1^{b}=\mathbf 1$ and
$D_{b+1}^{n'}=\mathbf 0$ hold, which imply that $Y = X_{b+1}^{n'}$.
Therefore the conditional entropy in (\ref{eqn:entropyexpansion}) is
nonzero only if $y = x_{b+1}^{n'}$. Step (g) holds because given $y
= x_{b+1}^{n'}$, the probability that $M_2=0$ is of order
$O(\beta^{2})$.

Define $l_{b}(\cdot): \{0,1\}^*\rightarrow \zZ^+$ to be the length
of the first $b$-run of $x^n$ (c.f. Definitions~\ref{def:brun} and
\ref{def:firstbrun}). In other words, for $l=1,2,\ldots$,
$l_{b}(x^{n}):=l$ if (i) $\forall b\leq i < b+l$, $x_i = x_{i-b}$
and (ii) $x_{b+l}\neq x_l$. Let $\mathbf{d}_{i,b}$ denote the
sequence $d_1^{n'}\in \{0,1\}^{n'}$ satisfying that if
$j=i,\ldots,i+b-1$, then $d_j=1$, otherwise $d_j=0$. Due to
Fact~\ref{fact:deletionbrun}, if $l_b(x^{n'}) = l$, then $y(x^{n'},
\mathbf{d}_{i,b}) = x_{b+1}^{n'}$ holds for all $i=1,\ldots,l$, but
does not hold for any $i>l$. Since given $D_0=0$ and $D_{n+1}=0$ all
$l$ deletion patterns $\{\mathbf{d}_{i,b} \}_{i=1}^l$ occurs with
the same probability $\alpha(1-\alpha)^{b-1}\beta (1-\beta)^{n'-b}$,
and only one of them, $\mathbf{d}_{1,b}$, satisfies $D_1=1$, we have
$H(D_1|X^{n'}=x^{n'},
Y=x_{b+1}^{n'},l(X^{n'})=l,(T,D_0,B)=(1,0,b))=h_2(1/l)$.

For a sequence $x^{n'}$ satisfying $l_b(x^{n'})=l$, we have
\begin{eqnarray*}
\lefteqn{p_{X^{n'},Y,T,D_0,B}(x^{n'},x_{b+1}^{n'},1,0,b)}\\
&=& p_{X^{n'}}(x^{n'}) \sum_{i=1}^l
p_{D_1^{n'}|D_0,D_{n'+1}}(\mathbf{d}_{i,b}|0,0)\\
&=& p_{X^{n'}}(x^{n'}) l
\alpha(1-\alpha)^{b-1}\beta (1-\beta)^{n'-b}\\
&=& p_{X^{n'}}(x^{n'}) l \alpha(1-\alpha)^{b-1}\beta
(1-O(\beta^{1-\epsilon})),
\end{eqnarray*}
for any $\epsilon>0$.

Therefore we continue (\ref{eqn:entropyexpansion2}) as
\begin{eqnarray*}
\lefteqn{\mbox{(\ref{eqn:entropyexpansion2})}}\\ &=&
\sum_{b=1}^{-k/3\log\beta} \sum_{l=1}^{n'-b}
\sum_{x^{n'}:l_b(x^{n'})=l} h_2\left(\frac{1}{l}\right)
p_{X^{n'}}(x^{n'}) l
\alpha(1-\alpha)^{b-1}\beta(1-O(\beta^{1-\epsilon}))\\
&&+O(\beta^2)\\
&=& \sum_{b=1}^{-k/3\log\beta} \sum_{l=1}^{n'-b}
h_2\left(\frac{1}{l}\right)2^{-l} l
\alpha(1-\alpha)^{b-1}\beta(1-O(\beta^{1-\epsilon}))+O(\beta^2)\\
&\stackrel{(h)}{=}& \sum_{b=1}^{\infty} \sum_{l=1}^{\infty}
h_2\left(\frac{1}{l}\right)2^{-l} l
\alpha(1-\alpha)^{b-1}\beta(1-O(\beta^{1-\epsilon}))+O(\beta^2)\\
&=& \sum_{l=1}^{\infty} h_2\left(\frac{1}{l}\right)2^{-l} l\beta+O(\beta^{2-\epsilon})\\
&\stackrel{(i)}{=}& C \beta+O(\beta^{2-\epsilon}),
\end{eqnarray*}
where step (h) holds because $k = \max \{6, 6/(\log (1-\alpha))\}$
and $n'= -2k/3 \log \beta$, which guarantee that changing the limits
of summations to infinity only leads to a change of order
$O(\beta^2)$, and step (i) holds because $\sum_{l=1}^{\infty}
h_2(1/l)2^{-l}l = \sum_{l=1}^{\infty} 2^{-l-1}l\log l$.
\end{proof}

\footnotesize

\bibliography{newbibfile}
\end{document}